\documentclass[a4paper,11pt]{article}

\title{Parameterized Convexity Testing\footnote{The work of the first author was supported by the Israel Science Foundation, grant number 592/17 and 822/18. The work of the second author was supported by the Israel Science Foundation, grant number 379/21.}}

\author{Abhiruk Lahiri\footnote{Ariel University, Israel, \texttt{abhiruk@ariel.ac.il}} \and Ilan Newman\footnote{University of Haifa, Israel, \texttt{ilan@cs.haifa.ac.il}} \and
Nithin Varma\footnote{Chennai Mathematical Institute, India, \texttt{nithinvarma@cmi.ac.in}}}

\date{ }
\usepackage{algorithm,algpseudocode}
\usepackage{mathrsfs,mathtools,xcolor,fullpage}
\usepackage{hyperref}
\usepackage{cleveref}
\usepackage{amsmath,amsthm,amsfonts,amssymb,enumerate}

\newcommand{\eps}{\varepsilon}
\newcommand{\ord}{\mathsf{order}}
\newcommand\restr[2]{{
  \left.\kern-\nulldelimiterspace#1\right|_{#2}
  }}

\newtheorem{theorem}{Theorem}[section]
\newtheorem{lemma}[theorem]{Lemma}
\newtheorem{definition}[theorem]{Definition}
\newtheorem{fact}[theorem]{Fact}
\newtheorem{observation}[theorem]{Observation}
\newtheorem{claim}[theorem]{Claim}
\newtheorem{remark}[theorem]{Remark}
\newcommand{\ignore}[1]{}

\begin{document}
\maketitle

\begin{abstract}
In this work, we develop new insights into the fundamental problem of convexity testing of real-valued functions over the domain $[n]$. 
Specifically, we present a nonadaptive algorithm that, given inputs $\eps \in (0,1), s \in \mathbb{N}$, and oracle access to a function, $\eps$-tests convexity in $O(\log (s)/\eps)$, where $s$ is an upper bound on the number of distinct discrete derivatives of the function.
We also show that this bound is tight.
Since $s \leq n$, our query complexity bound is at least as good as that of the optimal convexity tester (Ben Eliezer; ITCS 2019) with complexity $O(\frac{\log \eps n}{\eps})$; our bound is strictly better when $s = o(n)$. The main contribution of our work is to appropriately parameterize the complexity of convexity testing to circumvent the worst-case lower bound (Belovs et al.; SODA 2020) of $\Omega(\frac{\log (\eps n)}{\eps})$ expressed in terms of the input size and obtain a more efficient algorithm.
\end{abstract}

\section{Introduction}

A function $f:[n] \to \mathbb{R}$ is convex if $f(x) - f(x-1) \leq f(x+1) - f(x)$ for all $x \in \{2,3,\dots,n-1\}$.
Convexity of functions is a natural and interesting property.
Given oracle access to a function $f$, an $\eps$-tester for convexity has to decide with high constant probability, whether $f$
is a convex function or whether every convex function evaluates differently from $f$ on at least $\eps n$ domain points, where $\eps \in (0,1)$. 
Parnas, Ron, and Rubinfeld~\cite{ParnasRR06} gave an $\eps$-tester
for convexity that has query complexity $O(\frac{\log n}{\eps})$.
Blais, Raskhodnikova, and Yaroslavtsev~\cite{BlaisRY14} showed that this bound is tight for constant $\eps$ for nonadaptive algorithms\footnote{The queries of a nonadaptive algorithm does not depend on the answers to the previous queries. The algorithm is adaptive otherwise.}. An improved upper bound of $O(\frac{\log(\eps n)}{\eps})$ was shown by Ben-Eliezer~\cite{Ben-Eliezer19} in a work on the more general question of testing local properties. 
Recently, Belovs, Blais and Bommireddi~\cite{BelovsBB20} complemented this result by showing a tight lower bound of $\Omega(\frac{\log(\eps n)}{\eps})$.

In this work, we further investigate and develop new insights into this well-studied problem, thereby asserting that there is more way to go towards a full understanding of testing convexity of functions $f:[n] \to \mathbb{R}$. 
We show that the number of distinct discrete derivatives $s$, as opposed to the input size $n$, is the right input parameter to express the complexity of convexity testing, where a discrete derivative is a value of the form $f(x+1) - f(x)$ for $x \in [n-1]$.
Specifically, we design a nonadaptive convexity tester with query complexity $O(\frac{\log s}{\eps})$, and complement it with a nearly matching lower bound of $\Omega(\frac{\log (\eps s)}{\eps})$.
Our work is motivated by the work of Pallavoor, Raskhodnikova and Varma~\cite{PallavoorRV18} who introduced the notion of parameterization in the setting of sublinear algorithms. 

Our results bring out the fine-grained complexity of the problem of convexity testing.
In particular, $s \leq n$ always and therefore, our tester is at least as efficient as the state of the art convexity testers. 
Furthermore, the parameterization that we introduce, enables us to circumvent the worst case lower bounds expressed in terms of the input size $n$ and obtain more efficient algorithms when $s << n$.

\subsection{Our Results}

We begin our investigation with the simple and highly restricted case of testing convexity of functions $f:[n] \to \mathbb{R}$ having at most two distinct discrete derivatives. 
We design an adaptive algorithm that exactly decides convexity by making $5$ queries and a nonadaptive
algorithm that $\eps$-tests convexity by making $O(1/\eps)$ queries.
The highlight is that both these algorithms are deterministic.

\begin{theorem}\label{thm:convexity-2-deriv-test-adaptive}
There exists a \textbf{deterministic} algorithm that, given oracle access to a function $f:[n] \to \mathbb{R}$ having at most $2$ distinct discrete derivatives, exactly decides convexity by making at most $5$ adaptive queries.
\end{theorem}

Theorem~\ref{thm:convexity-2-deriv-test-adaptive} is significant because one can construct simple examples of two distributions, both over functions having at most $2$ distinct discrete derivatives, one over convex functions and the other over non-convex functions, such that no nonadaptive deterministic algorithm making $o(n)$ queries can distinguish the functions. Therefore, the above result shows the power of adaptivity even in this restricted setting.

We also design a constant-query deterministic nonadaptive testing algorithm for convexity of functions $f:[n] \to \mathbb{R}$ having at most $2$ distinct discrete derivatives.

\begin{theorem}\label{thm:convexity-2-deriv-test-nonadaptive}
Let $\varepsilon \in (0,1)$. There exists a \textbf{deterministic nonadaptive} $1$-sided error $\eps$-tester for convexity of functions $f:[n] \to \mathbb{R}$ having at most $2$ distinct discrete derivatives with query complexity $O(1/\varepsilon)$.
\end{theorem}

Next, we consider the general case of functions $f:[n] \to \mathbb{R}$ having at most $s$ distinct discrete derivatives and design the following nonadaptive tester.

\begin{theorem}\label{thm:convexity-r-deriv-test}
Let $\varepsilon \in (0,1)$. There exists a nonadaptive $1$-sided error $\varepsilon$-tester with query complexity $O(\frac{\log s}{\varepsilon})$ for convexity of real-valued functions $f:[n] \to \mathbb{R}$ having at most $s$ distinct discrete derivatives.
\end{theorem}

We complement Theorem~\ref{thm:convexity-r-deriv-test} with the following lower bound that is tight for constant $\eps \in (0,1)$. The bound holds even for adaptive testers, thereby showing that one cannot hope for a separation between adaptive and nonadaptive testers for this general setting.

\begin{theorem}\label{thm:lower-bound}
For every sufficiently large $s \in \mathbb{N}$, every $\eps \in [1/s, 1/9]$, and for every sufficiently large $n \ge s$, every $\eps$-tester for convexity of functions $f:[n] \to \mathbb{R}$ having at most $s$ distinct derivatives has query complexity $\Omega\left(\frac{\log (\eps s)}{\eps}\right)$.
\end{theorem}

\subsection{Related Work}
The study of property testing was initiated by Rubinfeld and Sudan~\cite{RubinfeldS96} and Goldreich, Goldwasser and Ron~\cite{GoldreichGR98}.
The first example where parameterization has helped in the design of efficient testers is the work of Jha and Raskhodnikova~\cite{JhaR13} on testing the Lipschitz property. A systematic study of parameterization in sublinear-time algorithms was initiated by Pallavoor, Raskhodnikova and Varma~\cite{PallavoorRV18} and studied further by~\cite{Belovs18,ChakrabartyS19,rameshThesis,NewmanV21}.  

In this work, we are concerned only with convexity of real-valued functions over a $1D$ domain.
We would like to note that not much is known about testing convexity of functions over higher dimensional domains. One possible reason behind this could be the following: there is no single definition of discrete convexity for real-valued functions of multiple variables. For a good overview of this topic, we refer interested readers to the textbook by Murota on discrete convex analysis~\cite{Murota}. 
Ben Eliezer~\cite{Ben-Eliezer19}, in his work on local properties, studied the problem of testing convexity of functions of the form $f:[n]^2 \to \mathbb{R}$ and designed a nonadaptive tester with query complexity $O(n)$. 
Later, Belovs, Blais and Bommireddi~\cite{BelovsBB20} showed a nonadaptive query lower bound of $\Omega(\frac{n}{d})^{\frac{d}{2}}$ for testing convexity of real-valued functions over $[n]^d$. For functions of the form $f:[3] \times [n] \to \mathbb{R}$, they design an adaptive tester with query complexity $O(\log^2 n)$ and show that the complexity of nonadaptive testing is $O(\sqrt{n})$.

\section{Preliminaries}
For a natural number $n \in \mathbb{N}$, we denote by $[n]$, the set $\{1,2,\dots, n\}$. 
Let $B \subseteq [n]$. For $x \in B$, the $\ord(x)$ is the number of points $y \in B$ such that $y \le x$. 
Points $x,y \in B$ are \emph{consecutive} if $|\ord(x) - \ord(y)| = 1$.
A function $f:B \to \mathbb{R}$ is \emph{convex} if and only if  
\begin{equation}
\label{eq:convex}
\frac{f(y) - f(x)}{y - x} \le \frac{f(z) - f(y)}{z - y}
\end{equation}
for all $x, y, z \in B$ such that $x < y < z$.
A set of points $V \subseteq B$ is said to \emph{violate} convexity if $\restr{f}{V}$ is not convex.

For a function $f:[n] \to \mathbb{R}$, and $i \in [n]$, we use the term `discrete derivative' at $i$ for $\Delta(f,i) = f(i+1) - f(i)$. We denote by $\Delta_f: [n-1] \to \mathbb{R}$, the derivative function $\Delta_f(i) = \Delta(f,i), ~ i=1, \ldots ,n-1$. 
The cardinality of the range of $\Delta_f$ is referred to as the \emph{number of distinct discrete derivatives} of $f$.
A function $f:[n] \to \mathbb{R}$ is convex if and only if $\Delta_f$ is monotone non-decreasing.

\begin{fact}
If $f:B \to \mathbb{R}$ is not convex, then there exists three consecutive points $x, y, z \in B$ that violate~\cref{eq:convex}.
\end{fact}

We  note that although convexity of $f$ is equivalent to monotonicity
of $\Delta_f$, it is not true that if $f$ is $\eps$-far from being
convex then $\Delta_f$ is $\nu$-far from monotonicity for some
positive constant $\nu$. E.g., consider $f$ that is defined by
$f(i)=i, ~ i\in [k]$, $f(k+1)=k$ and $f(j) = j-1, ~ j= k+2, \ldots
n$. Then $f$ if nearly $\frac{1}{2}$-far from being convex for $k =
  n/2$, while $\Delta_f$ is almost the $1$-constant function.

Let $\eps \in (0,1)$.
A function $f:[n] \to \mathbb{R}$ is $\eps$-far from convex if every convex function evaluates differently from $f$ on at least $\eps n$ points.
A \emph{basic $\eps$-tester} for convexity gets oracle access to a function $f$, a parameter $\eps$, and is such that, it
\textbf{accepts} if $f$ is convex,
    and \textbf{rejects}, with probability at least $\eps$, if $f$ is $\eps$-far from convex. 

\section{Deterministic Convexity Testers for Functions having at most $2$ Distinct Discrete Derivatives}
\label{sec:convextest}

We start with the very simple case in which $f$ has only $2$ distinct
derivatives (if $f$ has only $1$ distinct derivative then $f$ is
a degree $1$ function and, in particular, convex) and prove Theorem~\ref{thm:convexity-2-deriv-test-adaptive} and Theorem~\ref{thm:convexity-2-deriv-test-nonadaptive}.

 Let the range of $\Delta_f$
be $\{r_1 < r_2\}$. We do not assume that the algorithm knows the values $r_1,
r_2$ but we will refer to them in our proofs and reasoning below. As
it turns out, in this case there is a deterministic adaptive algorithm
that can precisely decide if $f$ is convex, making only $5$
queries. This is based on the fact that the class of convex functions
having at most two distinct derivatives is very restricted.

  \begin{observation}
    \label{obs:0.2}
    If $f$ is convex and $\Delta_f$ takes at most two values $\{r_1 < r_2\}$ then $f$ is
    of the following form: there is $j \in [n]$ such that  $f(i) =a +
    r_1 (i-1), i=1, \ldots , j$ and $f(i) = a + r_1(j-1) + r_2 (i-j),~
    i=j+1, \ldots n$ for some $r_1 < r_2$ and some $a$.  
    We denote such $f$ as $f_{a,r_1,j,r_2}$.
  \end{observation}

  \ignore{
  \label{obs:1}
  Let $f$ be $\eps$-far from being convex and having at most $2$
  distinct derivatives, then there is a $j \in [\frac{\eps n}{2},
  (1- \frac{\epsilon}{2})n]$ such that $\Delta_f(j+1) < \Delta_f(j)$.
\end{observation}
Namely, the observation claims that there a violation to the
monotonicity of $\Delta_f$ in the interval $[\frac{\eps n}{2},
  (1- \frac{\eps}{2})n]$.
\begin{proof}
Let $I = [a,b]$ for $a \leq \frac{\eps n}{2}$ and $b \geq (1- \frac{\eps}{2})n,$ and assume, for the contrary, that for any $j \in I,$
 $\Delta_f(j+1) \geq \Delta_f(j)$. Define
  $f'(j) = f(j), ~ i \in I$ and $f'(i) = f(a), i \leq a$ and $f'(i) =
  f(b) + (i-b)\Delta_f(b-1)$. It is easy to see that $f'$ is convex
  while it is $\epsilon$-close to $f$.  
\end{proof}
}

Observation \ref{obs:0.2} suggests Algorithm~\ref{alg:determine-test} as a test for convexity.
\medskip 

\begin{algorithm}
\begin{algorithmic}[1]
\Require oracle access to function $f:[n] \to \mathbb{R}$ having at most $2$ distinct derivatives
\State Query $f(1), f(2), f(n-1), f(n)$.
\State Define $f_1(x) \coloneqq f(1) + (x-1)(f(2)- f(1))$ and $f_2(x) \coloneqq f(n) - (n-x)(f(n)-f(n-1))$.
\State Query $f(j)$ such that $f_1(j) = f_2(j)$.
\State \textbf{Reject} if the function restricted to $1, 2, j, n-1, n$ is not convex or $f(j) \neq f_1(j)$, and \textbf{accept} otherwise.
\end{algorithmic}
\label{alg:determine-test}
\caption{}
\end{algorithm}

\ignore{{\bf Test 1:}
Query $f(1), f(2), f(n), f(n-1)$ and also $f(j)$ for $j$ such that
$f_1(j) = f_2(j)$, where $f_1(x) = f(1) + (f(2)- f(1)) (x-1)$ and
$f_2(x) = f(n) - (f(n)-f(n-1)) (n-x)$. 
Reject if there is no such integer $j$, or if $f$ restricted to
the queried points is not convex, and otherwise accept.}

\begin{remark}
Algorithm~\ref{alg:determine-test} is a \emph{deterministic, adaptive} algorithm; it can make the last
query only after knowing the values of $f$ at the first $4$ points. 
Moreover, if $f$ is a function having at most $2$ distinct discrete derivatives, then there is always an integer point $j \in [n]$ such that $f_1(j) = f_2(j)$.
\end{remark}

\begin{lemma}
Algorithm~\ref{alg:determine-test} accepts every convex function having at most $2$ distinct
derivatives, and rejects every function having at most $2$ distinct
derivatives that is not convex. 
\end{lemma}
We note that the lemma asserts that Algorithm~\ref{alg:determine-test} decides convexity correctly
on every function that has at most $2$ distinct derivatives,
regardless of the distance to convexity.

\begin{proof}
If $f$ is convex, then the restriction of $f$ to every subset of $[n]$ is also convex and Algorithm~\ref{alg:determine-test} accepts.

Suppose $f$ is not convex. 
This immediately implies that $f$ has two distinct discrete derivatives, which we denote by $r_1 < r_2$. 
Now, it is necessary that $r_1 = \Delta_f(1)$ and $r_2 =
\Delta_f(n-1)$ for the restriction of $f$ to $\{1,2, n-1, n\}$ to be convex, for 
otherwise Algorithm~\ref{alg:determine-test} immediately rejects.

\ignore{Let $\Delta_f(1) = \alpha$. Note that if $\alpha= r_2$ (that is, the
largest derivative), then $f$ is
not convex, iff $f(n) < f(1) + \alpha (n-1)$.
In this case $1,2,n$ form a violation to convexity. 
Hence we may
assume that $\alpha=r_1$ (the min-value derivative)  and that $f(n) \geq f(1) + \alpha (n-1)$.

Let $\beta = f(n) - f(n-1)$. An analogous reasoning to the above implies
that we may assume that $\beta=r_2$ and that $f(1) > f(n) - \beta
(n-1)$ as otherwise $1,n-1,n$ form a violation to convexity. So, we
conclude that we know .}
Assuming that the restriction of $f$ to the set
$S = \{1,2,n-1, n \}$ is convex, Observation \ref{obs:0.2} implies that the only convex
function with $2$ distinct derivatives that is consistent with $f$ on the points in $S$ is the function
$f^*$ (from Observation~\ref{obs:0.2}), where the value of $j \in [n]$ is unique and is as determined by Algorithm~\ref{alg:determine-test}. 
If $f(j) \neq f^*(j)$, then the restriction of $f$ to $S \cup \{j\}$ is not convex 
and Algorithm~\ref{alg:determine-test} rejects.
In the rest of the proof, we argue that if $f(j) = f^*(j)$, then $f$ and $f^*$ has to evaluate to 
the same value on every point in $[n]$ and that $f$ is convex.
Specifically, each one of the discrete derivatives of $f$ upto the $j$ must be $r_1$, for 
otherwise, $f(j)$ will be larger than $f^*(j)$. Moreover,
each one of the discrete derivatives of $f$ from $j$ upto $n$ must be $r_2$, for otherwise,
$f(j)$ will be smaller than $f^*(j)$. 
That is, the functions are identical on every point in $[n]$. 
\end{proof}


\subsection{A Nonadaptive Deterministic Convexity Tester for Functions having at most $2$ Distinct Discrete Derivatives}

As remarked above, Algorithm~\ref{alg:determine-test} is deterministic and decides convexity
exactly under the promise that $f$ has at most $2$ distinct derivatives.
However, it is \emph{adaptive}. What can be said about nonadaptive
algorithms for the same problem? It is easy to see that for any deterministic algorithm
that makes $q < n - 1$ \emph{nonadaptive} queries $Q
\subset [n]$, there
are two functions $g,h$, both having at most $2$ distinct derivatives,
for which $g|_Q = h|_Q$ but $g$ is convex while $h$ is not
convex. Hence there is no deterministic nonadaptive algorithm that decides
convexity exactly, while making at most $n-2$ queries. This line of reasoning immediately
extends to a $\Omega(n)$ lower bound on randomized nonadaptive algorithms
that exactly decide convexity.

Here we come back to the property testing
scenario. We show that there is a \emph{deterministic nonadaptive}
tester, Algorithm~\ref{alg:nonadaptive-test}, that accepts every convex function $f$ and rejects
every function $f$ that is $\eps$-far from convex, under the promise that
$f$ has at most $2$ distinct derivatives.  Algorithm~\ref{alg:nonadaptive-test} makes only
$O(1/\eps)$ nonadaptive deterministic queries.

\begin{algorithm}
\caption{}
\begin{algorithmic}[1]
\Require $\eps \in (0,1)$; oracle access to function $f:[n] \to \mathbb{R}$ having at most $2$ discrete derivatives
\State $S \gets \{x_i (= i \eps n )|  ~i=1, \ldots , 1/\eps \}$
\State Query $f(1), f(2), f(n-1), f(n)$ and $f(x_i),f(x_i+1)$ for all $i=1, \ldots, 1/\eps$
\State Set $r_1  \leftarrow f(2)-f(1)$ and $r_2 \leftarrow f(n)-f(n-1)$
\State \textbf{Reject} if $r_1 > r_2$
\State Let $j \in [1/\eps]$ be the largest integer such that $f(x_j) = f(1) + r_1 (x_j-1)$
\State \textbf{Reject} if for some $i \geq j+1$, $f(x_i) \neq f(n) - r_2(n-x_i)$ and \textbf{accept} otherwise
\end{algorithmic}
\label{alg:nonadaptive-test}
\end{algorithm}

\ignore{
{\bf Test 2:}  Choose $3+1/\eps$ points $\{x_i = i \eps n, ~i=1, \ldots , 1/\eps \}$ equally spaced in $[n]$. Query
$f(1),f(2),f(n), f(n-1),$
and $f(x_i),f(x_i+1)$ for all $i=1, \ldots, 1/\eps$. Let $r_1  =
f(2)-f(1)$ and $r_2 = f(n)-f(n-1)$. Reject if $r_1 > r_2$.

Let $j \in [1/\eps]$ be the largest such that $f(x_j) = f(1) + r_1 (x_j-1)$. 
Reject if for some $i \geq j+1$ $f(x_i) \neq f(n) - r_2
(n-x_i)$. Otherwise accept.}

\begin{claim}
  Algorithm~\ref{alg:nonadaptive-test} accepts every convex function $f$ and rejects every
  function $f$ that is $\eps$-far from convex, provided that $f$
  has at most $2$ distinct derivatives. Further, Algorithm~\ref{alg:nonadaptive-test} is a deterministic
  nonadaptive tester making $O(1/\eps)$ queries.
\end{claim}
\begin{proof}
  The claim about the query complexity is clear. Further, by
  Observation \ref{obs:0.2}, if $f$ is convex having at most $2$
  distinct derivatives, then for some $j^* \in [n]$, the function $f$ is of the 
  form given in the observation. Let $x_j \leq j^* \leq x_{j+1}$. Then $f$ is consistent
  with $j$ in the acceptance criterion of Algorithm~\ref{alg:nonadaptive-test}, and hence will be
  accepted.

Next, consider a function $f$ that is accepted by Algorithm~\ref{alg:nonadaptive-test}. 
That is, there exists $j \in [1/\eps]$ such that 
$f(x_i) = f(1) + r_1 (x_i-1)$ for every $i \leq
j$, and that $f(x_i) = f(n) - r_2(n-x_i)$ for every $i \geq
j + 1$, where $r_1 \leq r_2$ are the two distinct discrete derivatives of $f$.
Since $f(x_j) = f(1) + r_1 (x_j-1)$, the function $f$ is a linear 
function when restricted to the set $[x_j]$. 
Similarly, when restricted to the set $[n] \setminus [x_{j+1} - 1]$,
the function $f$ is linear with slope $r_2$.
Further, it can be seen that $f$ can be corrected to be
convex by changing the values for $x_{j}+1 \leq i \leq x_{j+1}-1$ to
be consistent with $f_{f(1),r_1,j^*, r_2}$. As this changes at most 
$\eps n$ points, it implies that $f$ is $\eps$-close to convex.
\end{proof}

\section{Convexity Tester for Functions having at most $s$ Distinct Discrete Derivatives}

In this section, we describe our convexity tester for the case that the function $f:[n] \to \mathbb{R}$ has at most $s$ distinct discrete derivatives and prove Theorem~\ref{thm:convexity-r-deriv-test}.
A basic tester is presented in Algorithm~\ref{alg:convexity-r-derivative-tester}. 
For simplicity, we assume throughout this section, that $s/\eps$ is an integer that divides $n$.

The top level idea is the following: suppose that $f$ is convex with
at most $s$ distinct discrete 
derivatives, and let  $B$ be a set of $\ell =
1+\frac{2s}{\eps}$ nearly equally spaced consecutive pairs of
points in $[n]$ starting with $1,2$, namely, $B = \{1, 2\} \cup 
\{i\cdot \frac{\eps n}{2s} - 1, i\cdot \frac{\eps n}{2s}: ~ i=1, \ldots \ell-1\}$. 
Let $x_i = i\cdot \frac{\eps n}{2s}$ for $i \in [\ell - 1]$. By the
assumption on $f$, the function $\Delta_f|_I$ is the constant function on  at least
$\ell - s$ of the intervals $I = [x_i, x_{i+1}-1]$. Further, if
$\Delta_f|_I$ is constant on $I = [x_i, x_{i+1}-1]$, then obviously
$f(j) = f(x_i) + (j-x_i) \cdot \Delta_f(x_i)$ for $j \in I$.  Thus in order to check
that $f$ is convex, we first check that $f|_B$ is convex using the
nonadaptive, 1-sided
error basic $\eps$-tester of Belovs et al. \cite{BelovsBB20} by making $O(\log (\eps|B|))$ queries. Afterwards, we test that $f$ is close to being a
linear function on most intervals
$I$.  To test
``linearity'' of $f|_I$ on most such $I$, it is enough to pick a random such
interval and test the distance to the appropriate linear function, which will result in a large enough success probability. The details follow.

Algorithm~\ref{alg:convexity-r-derivative-tester} invokes Algorithm~\ref{alg:pseudo-er-POT-convexity} as a subroutine, where Algorithm~\ref{alg:pseudo-er-POT-convexity} is a basic tester for convexity of functions defined over subdomains of $[n]$.
The following theorem can be proven by modifying the analysis of the convexity tester by Belovs et al.~\cite{BelovsBB20} in a fairly straightforward manner. 
We have included its proof in the Appendix.
\begin{theorem}[Belovs et al.~\cite{BelovsBB20}]\label{thm:belovs-pseudo-er-POT}
Let $B \subseteq [n]$.
There exists a basic $\eps$-tester for convexity of functions of the form $f:B \to \mathbb{R}$ that works for all $\eps \in (0,1)$ with query complexity $O(\log(\eps |B|))$.
\end{theorem}

\begin{algorithm}
\caption{Convexity Tester}
\begin{algorithmic}[1]
\Require parameter $\eps \in (0,1)$; oracle access to function $f:[n] \to \mathbb{R}$; upper bound $s$ on the number of distinct discrete derivatives in $f$
\State Let $B = 
\{1, 2\}\cup \{i\cdot \frac{\eps n}{2s} - 1, i\cdot \frac{\eps n}{2s} : ~ i=1, \ldots \frac{2s}{\eps}\}$. 
\State Test convexity of $f|_B$ with parameter $\eps/32$ using Algorithm~\ref{alg:pseudo-er-POT-convexity} and
\textbf{reject} if that execution rejects.
\label{step:pseudo-er-convexity-test}
\State Sample a point $x \in_R [n]$ u.a.r.
\State Let $y \gets \left\lfloor\frac{2sx}{\eps n}\right\rfloor$. 
\State \textbf{Reject} if the points $y-1, y, x, y-1+\frac{\eps n}{2s}, y+\frac{\eps n}{2s}$ violate convexity.
\label{step:convexity-test-loop-end}
\end{algorithmic}
\label{alg:convexity-r-derivative-tester}
\end{algorithm}

\begin{algorithm}[t]
\caption{Basic Tester for Convexity over Subdomains of $[n]$}
\begin{algorithmic}[1]
\Require oracle access to a function $g:B \to \mathbb{R}$; parameter $\eps$
\Loop~$\lceil 24 \log(2\eps|B|)\rceil$~times:
\State Draw a point $a \in B$ uniformly at random.
\State Let $a' \in B$ be such that $\ord(a') = \ord(a)+1$.
\State Pick a number $k$ uniformly at random from $\{0, 1, \dots, \lceil \log 2\eps |B| \rceil\}$.
\State Let $h \in B$ be a point such that $\mathsf{order}(h)$ is a multiple of $2^k$, where, with probability $1/2$, the point $h$ is the smallest such point larger than $a$, and with probability $1/2$, it is the largest such point smaller than $a$.
\State Let $h' \in B$ be such that $\mathsf{order}(h') = \mathsf{order}(h) + 1$.
\State \textbf{Reject} if the set $\{a,a',h, h'\}$ violates convexity.
\EndLoop
\end{algorithmic}
\label{alg:pseudo-er-POT-convexity}
\end{algorithm}

Lemma~\ref{lem:conv-general-analysis} shows that Algorithm~\ref{alg:convexity-r-derivative-tester} is indeed a basic $\eps$-tester for convexity.
Our convexity tester with query complexity $O(\log(s)/\eps)$ is obtained by $O(1/\eps)$ repetitions of Algorithm~\ref{alg:convexity-r-derivative-tester}.
This completes the proof of Theorem~\ref{thm:convexity-r-deriv-test}.

\begin{lemma}\label{lem:conv-general-analysis}
Algorithm~\ref{alg:convexity-r-derivative-tester}, rejects with probability at least $\eps/32$, every function $f$ having at most $s$ distinct discrete derivatives that is $\eps$-far from convex.
\end{lemma}
\begin{proof}
It is enough to argue that Steps~\ref{step:pseudo-er-convexity-test}-\ref{step:convexity-test-loop-end} of Algorithm~\ref{alg:convexity-r-derivative-tester} rejects with probability at least $\eps/32$. 

If $f|_B$ is $\eps/32$-far from being convex, by Theorem~\ref{thm:belovs-pseudo-er-POT}, one iteration of Algorithm~\ref{alg:pseudo-er-POT-convexity} rejects with probability at least $\eps/32$. 

In the rest of the proof, we assume that $f|_B$ is $\eps/32$-close to convex. 
In other words, it is possible to modify $f|_B$ in at most $\frac{\eps|B|}{32}$ points in $B$ in order to make $f|_B$ convex.

Let $b_k$ for $k \in [2s/\eps]$ be shorthand for the index $k \cdot \frac{\eps n}{2s} - 1$, and let $b_0$ stand for the index $1$. 
Let $I_k$ denote the interval of indices $\{b_{k} + 1, \dots , b_{k+1}\}$ for $k \in [(2s/\eps) - 1]$. 
Let $I_0$ denote the interval of indices $\{2, \dots, b_1\}$.
For $k \in [(2s/\eps) - 1]$, the interval $I_k$ is \emph{nearly linear} if 
\begin{equation}
f(b_k+1) - f(b_k) = \frac{f(b_{k+1}) - f(b_k + 1)}{b_{k+1} - (b_k + 1)} = f(b_{k+1}+1) - f(b_{k+1}).
\label{equn:slope-general}
\end{equation}
The interval $I_0$ is nearly linear if 
\begin{equation}
f(2) - f(1) = \frac{f(b_1) - f(1)}{b_1 - 1} = f(b_1+1) - f(b_1). 
\label{equn:slope-first-interval}
\end{equation}

We first prove a lower bound on the number of nearly linear intervals. 
Recall that there is a way to modify $f|_B$ by changing its values on at most $\eps \cdot |B|/32$ \emph{bad} points. 
For $k \in [(2s/\eps) - 1]$, the interval $I_k$ is bad if one among $b_k, b_k+1, b_{k+1}, b_{k+1} + 1$ is a bad point, and is good otherwise.
Likewise, $I_0$ is bad if one among $1, 2, b_1, b_1 + 1$ is a bad point, and is good otherwise.
The number of bad intervals is, therefore, at most $\eps |B|/16$, since a bad point can make at most two intervals bad. Now, $\eps|B|/16$ is at most $15 s/16$, by substituting the value of $|B|$. 
 
For $k \in [(2s/\eps) - 1]$, if the interval $I_k$ is good, then none of the points in $\{b_k, b_k + 1, b_{k+1}, b_{k+1} + 1\}$ are bad, and hence we have
\begin{equation}
f(b_k+1) - f(b_k) \le \frac{f(b_{k+1}) - f(b_k + 1)}{b_{k+1} - (b_k + 1)} \le f(b_{k+1}+1) - f(b_{k+1}).
\label{equn:slope-general-nondecreasing}
\end{equation}

Similarly, if $I_0$ is good, then 
\begin{equation}
f(2) - f(1) \le \frac{f(b_1) - f(1)}{b_1 - 1} \le f(b_1+1) - f(b_1). 
\label{equn:slope-first-interval-nondecreasing}
\end{equation}

For a good interval $I_k$ (or $I_0$) for $k \in [(2s/\eps) - 1]$ that is not nearly linear, one of the inequalities in Equation~\ref{equn:slope-general-nondecreasing} (Equation~\ref{equn:slope-first-interval-nondecreasing}, respectively) must be a strict inequality. 
Since the number of distinct discrete derivatives in $f$ is at most $s$, the number of distinct discrete derivatives among restricted to the good points is also at most $s$. Since the function $f$ restricted to the good points is convex, the number of good intervals with strict inequalities (in Equation~\ref{equn:slope-general-nondecreasing}  or Equation~\ref{equn:slope-first-interval-nondecreasing}) is at most $s-1$. Hence, the number of good intervals that are not nearly linear is at most $s-1$. 

Since $f$ is $\eps$-far from being convex and each interval has at most $\frac{\eps n}{2s}$ indices, the restriction of $f$ to the set of indices belonging to good nearly linear intervals, has distance at least $\eps n - \frac{15s}{16} \cdot \frac{\eps n}{2s} - s \cdot \frac{\eps n}{2s} = \eps n - \frac{\eps n}{2} - \frac{15\eps n}{32}$ from convexity.



For $k \in [(2s/\eps) - 1]$, for a good nearly linear interval $I_k$, we use $D_k$ to denote
the Hamming distance of $\restr{f}{I_k}$ to the linear function $g_k: I_k \to \mathbb{R}$ defined as 
$g_k(x) = f(b_{k} + 1) + (x - b_k - 1)\cdot t$ for $x \in I_k$, where  $t = f(b_{k+1} + 1) - f(b_{k+1})$.
Similarly, if $I_0$ is good, we use $D_0$ to denote the Hamming distance of $\restr{f}{I_0}$ to the linear function $g_0: I_0 \to \mathbb{R}$ defined as $g_0(x) = f(1) + (x - 1)\cdot t$, where $t = f(2) - f(1)$.

Consider the restriction of $f$ to the set of indices that belong to the good nearly linear intervals. We can make this restriction convex by replacing $\restr{f}{I_k}$ with $g_k$ for each $k$ such that $I_k$ is a good nearly linear interval. Hence, \[\sum_{\substack{k: I_k \\ \text{nearly linear}\\ \text{and good}}} D_k \ge \eps n - \frac{\eps n}{2} - \frac{15\eps n}{32} \ge \frac{\eps n}{32}.\]

The proof will be completed by arguing that there are at least $D_k$ points $x$ in a good nearly linear interval $I_k$ such that Algorithm~\ref{alg:convexity-r-derivative-tester} rejects by sampling $x$ in Step~\ref{step:convexity-test-loop-end}.

Consider a good nearly linear interval $I_k$ such that $f(b_k +1) - f(b_k) = t$. 
Consider the (favorable) set $F$ consisting of all points $z \in I_k$ such that $f(z) - f(b_{k} + 1) = t \cdot (z - b_{k} - 1)$ and $f(b_{k+1}) - f(z) = t \cdot (b_{k+1} - z)$. 
Clearly, both $b_k+1$ and $b_{k+1}$ are in $F$.
If Algorithm~\ref{alg:convexity-r-derivative-tester} samples, in Step~\ref{step:convexity-test-loop-end}, a point $z \notin F$, it rejects.
We now show that we can repair the function values at points not in $F$ and make $\restr{f}{I_k}$ be equal to $g_k$. 
Consider an interval of points $\{x, x+1, \dots, y\}$ such that none of them are in $F$, and where, both $x-1$ and $y+1$ are in $F$. 
Since $x-1$ and $y+1$ are both in $F$, we have that $f(y+1) - f(x-1) = t \cdot (y - x +2)$. 
We repair the function on the interval $\{x, x+1, \dots, y\}$ by assigning the value $f(x-1) + (x'- x +1) \cdot t$ for all $x' \in \{x, x+1, \dots, y\}$. 
We can repair the function on the whole interval and make $\restr{f}{I_k}$ be equal to $g_k$ by applying the same modification on every such maximal subinterval, where the maximality is in the sense of not belonging to $F$. 

Hence, the probability that the tester rejects in 
Step~\ref{step:convexity-test-loop-end} is at least $\eps/32$. 
This completes the proof.
\end{proof}

\section{Lower Bound}
In this section, we prove Theorem~\ref{thm:lower-bound}.
\begin{lemma}
For every sufficiently large $s \in \mathbb{N}$, every $\eps \in [1/s, 1/9]$, and for every sufficiently large $n \ge s$, every $\eps$-tester for convexity of functions $f:[n] \to \mathbb{R}$ having at most $s$ distinct derivatives has query complexity $\Omega\left(\frac{\log (\eps s)}{\eps}\right)$.
\end{lemma}
\begin{proof}
We use Yao's principle.
Let $s \in \mathbb{N}$ and $\eps \in [1/s,1/9]$.
Consider the distributions $\mathcal{D}_0$ and $\mathcal{D}_1$ from Belovs et al.~\cite{BelovsBB20} (proof of Theorem 1.3) of functions $f: [s] \to \mathbb{R}$.
Every function sampled from these distributions have at most $s$ distinct derivatives.
Moreover, every function sampled from $\mathcal{D}_0$ is convex, and every function sampled from $\mathcal{D}_1$ is $\eps$-far from convex.
They show that every tester distinguishing these distributions, with probability at least $2/3$, has to make at least $\Omega\left(\frac{\log (\eps s)}{\eps}\right)$ queries.

Consider an integer $n \ge s$ that is an integer multiple of $s$.
Let $k$ denote $n/s$.
We define distributions $\mathcal{D}'_0$ and $\mathcal{D}'_1$ of functions $f: [n] \to \mathbb{R}$ as follows.

For $b \in \{0,1\}$, to sample a function $f$ from $\mathcal{D}'_b$, first sample a function $g$ from $\mathcal{D}_b$.
For $i \in [s]$, let $f((i-1)k + 1) = g(i)$.
For $i \in [s-1]$, let $\mathsf{slope}_i = \frac{g(i+1) - g(i)}{k}$.  
Now, for all $j \in [k-1]$ and for all $i \in [s-1]$, set $f((i-1)k+1 + j) = f((i-1)k+1) + j \cdot \mathsf{slope}_i$.

By construction, every function sampled from $\mathcal{D}'_0$ is convex. Additionally, every function sampled from $\mathcal{D}'_1$ is $\eps$-far from convex. 
To see this, consider a function $f$ sampled from $\mathcal{D}'_1$ that is $\eps$-close to being convex. Let $g$ denote the function sampled from $\mathcal{D}_1$ from which we constructed $f$. 
Let $B \subseteq [n]$ denote the set of \emph{bad} points such that changing the values of $f$
on points in $B$ makes it convex. 
Since $f$ is piecewise linear, it is clear that for each point in $B$ of the form $(i-1)k+1$ for $i \in [s]$, either
the set of $k-1$ points $\{(i-1)k+1 + j: j \in [k-1]\}$, or the set of $k-1$ points $\{(i-2)k+1 + j: j \in [k-1]\}$ has to belong to $B$.
Thus, the number of points in $B$ of the form $(i-1)k+1$ for $i \in [s]$ is at most $|B|/k = \eps s$.
By construction of $f$, we know that for all $i \in [s]$, it holds that $f((i-1)k + 1) = g(i)$.
Thus, the distance of $g$ to convexity is at most $\eps s$.

Consider a deterministic algorithm $A$ that distinguishes these distributions by making $o\left(\frac{\log (\eps s)}{\eps}\right)$ queries. 
One can use $A$ to distinguish, with the same success probability, the distributions $\mathcal{D}_0$ and $\mathcal{D}_1$ by making at most twice the number of queries as $A$, which leads to a contradiction.
\end{proof}

\bibliography{convex_test_arxiv}
\newpage
\begin{appendix}
\section{Basic Tester for Convexity over Subdomains of $[n]$}
In this section, we prove Theorem~\ref{thm:belovs-pseudo-er-POT}. 
All definitions and lemmas in this section are straightforward generalizations of those by Belovs et al.~\cite{BelovsBB20}.

\begin{definition}
[Test-Set, and its Root, Hub, and Scale]
Consider a point $a \in B$ and an integer $k$. Let $h \in B$ be one among the two points closest to $a$ such that $\ord(h)$ is a multiple of $2^k$. 
Let $a', h' \in B$ be such that $\ord(a') = \ord(a) + 1$, and $\ord(h') = \ord(h) + 1$.
We refer to $\{a,a',h,h'\}$ as the \emph{test-set} with \emph{root} $a$, \emph{hub} $h$, and \emph{scale} $2^k$. 
\label{def:test-sets}
\end{definition}

\begin{lemma}\label{lem:test-triplet-common-hub}
Consider points $x, y \in B$ such that $\ord(x) < \ord(y) - 1$. Then there exists test-sets with roots $x$ and $y$ respectively, having a common hub and scales at most $2(\ord(y) - \ord(x)).$ 
\end{lemma}
\begin{proof}
Consider the smallest integer $k$ such that there exists a unique multiple $m$ of $2^k$ satisfying $\ord(x) \le m \le \ord(y)$.
There are at least $2$ multiples of $2^{k-1}$ in the range $[\ord(x),\ord(y)]$.
If there were only one multiple of $2^{k-1}$ in that range, then it contradicts our assumption that $k$ is the smallest integer such that there is a unique multiple of $2^k$ in $[\ord(x),\ord(y)]$. 
Now, since there are at least $2$ multiples of $2^{k-1}$ in the range $[\ord(x),\ord(y)]$, we have that $2^{k-1} \le (\ord(y) -\ord(x))$.

The lemma follows by setting $h\in B$ such that $\ord(h) = m$ as the common hub of the test-sets with roots $x$ and $y$, and $2^k$ as their scale.
\end{proof}
\begin{lemma}\label{lem:conv-violation}
Consider a set $\{x,y,z\} \subseteq B$ that violates convexity such that $\ord(x) < \ord(y) < \ord(z)$. 
At least one of the test-sets involving $x,y,$ or $z$ violate convexity. Moreover, scale is not exceeding $2\cdot\max\{\ord(z)-\ord(y), \ord(y)-\ord(x)\}$.
\end{lemma}
\begin{proof}
Let $x', y', z' \in B$ be such that $\ord(x') = \ord(x) + 1$, $\ord(y') = \ord(y) + 1$, and $\ord(z') = \ord(z) + 1$. 

Let $h_1$ be the common hub for test-sets involving $x$ and $y$, as guaranteed by Lemma~\ref{lem:test-triplet-common-hub} and let $h_1' \in B$ be such that $\ord(h_1') = \ord(h_1) + 1$.
By Lemma~\ref{lem:test-triplet-common-hub}, the test-sets $\{x,x',h_1,h_1'\}$, and $\{h_1, h_1',y,y'\}$ both have scale at most $2(\ord(y)-\ord(x))$.

Similarly, let $h_2 \in B$ be the common hub for test-sets of $y$ and $z$ with scale at most $2(\ord(z)-\ord(y))$ and let $h_2' \in B$ be such that $\ord(h_2') = \ord(h_2)+1$.  
Note that $\{y,y', h_2, h_2'\}$, and $\{h_2, h_2', z, z'\}$ are the test-sets being alluded to in this case.

If none of the aforementioned four test-sets violate convexity, then, we immediately have the following inequalities:

\begin{align*}
    \frac{f(y) - f(x)}{y - x} \le \frac{f(y) - f(h_1')}{y - h_1'} \le f(y') - f(y) \le \frac{f(z) - f(y)}{z - y}.
\end{align*}
This contradicts our assumption that $\{x,y,z\}$ violates convexity.
\end{proof}
We are now ready to prove Theorem~\ref{thm:belovs-pseudo-er-POT}.
The query complexity of Algorithm~\ref{alg:pseudo-er-POT-convexity} is clear from its description. Additionally, it always accepts convex functions.
Consider a function $f:B\to \mathbb{R} $ that is $\eps$-far from convex. 
The total number of possible scales for test-sets (see Algorithm~\ref{alg:pseudo-er-POT-convexity}) is at most $1 + \lceil \log_2 (2\eps|B|)\rceil$.
Hence, the total number of test-sets is at most $2|B|\cdot (1 + \lceil \log_2 (2\eps|B|)\rceil)$.

We will construct a set $C \subseteq B$ such that $|C| \ge 
\eps \cdot |B|$ and for every $a \in C$, one of the test-sets rooted at $a$ of scale at most $1 + \lceil \log_2 (2\eps|B|)\rceil$ violate convexity.
Initialize $C$ to be $\emptyset$. 
If $|C| < \eps \cdot |B|$, then $\restr{f}{B\setminus C}$ is not convex.
We can find consecutive points $x < y < z \in B \setminus C$, such that $f$ violates convexity on these three points. 
Since $C$ contains fewer than $\eps |B|$ many points, we have that $\max\{\mathsf{order}(y) - \ord(x), \ord(z) - \ord(y)\}$ is at most $\eps |B|$.
Hence, by Lemma~\ref{lem:conv-violation}, we know that there exists a violating test-set involving either $x,y,$ or $z$ of scale at most $2 \eps |B|$.
We add such a point to set $C$.
Hence, we have at least $\eps n$ violating test-sets.
Therefore, the tester rejects with probability at least $\frac{2\eps}{\log(2\eps n)}$ in a single iteration.
\end{appendix}

\end{document}